\documentclass[conference]{IEEEtran}
\IEEEoverridecommandlockouts
\usepackage[letterpaper, left=.63in, right=.63in, top=.75in, bottom=1in]{geometry}
\usepackage{cite}
\usepackage{amsmath,amssymb,amsfonts}
\usepackage{algorithmic}
\usepackage{graphicx}
\usepackage{textcomp}
\usepackage{xcolor}
\usepackage{romannum}
\usepackage{amsthm}
\usepackage{algorithmic}
\usepackage[ruled,vlined,linesnumbered]{algorithm2e}
\usepackage{float}
\usepackage{subfigure}
\SetAlFnt{\footnotesize}
\def\BibTeX{{\rm B\kern-.05em{\sc i\kern-.025emF b}\kern-.08em
T\kern-.1667em\lower.7ex\hbox{E}\kern-.125emX}}

\newtheorem{lemma}{Lemma}

\begin{document}

\title{QF-MAC: Adaptive, Local Channel Hopping \\  for Interference Avoidance in Wireless Meshes}

\author{\IEEEauthorblockN{Yung-Fu Chen}
\IEEEauthorblockA{
\textit{The Ohio State University}\\
Columbus, OH, USA \\
chen.6655@osu.edu}
\and
\IEEEauthorblockN{Anish Arora}
\IEEEauthorblockA{
\textit{The Ohio State University}\\
Columbus, OH, USA \\
anish@cse.ohio-state.edu}
}


\maketitle

\begin{abstract}
The throughput efficiency of a wireless mesh network with potentially malicious external or internal interference can be significantly improved by equipping routers with multi-radio access over multiple channels. For reliably mitigating the effect of interference, frequency diversity (e.g., channel hopping) and time diversity (e.g., carrier sense multiple access) are conventionally leveraged to schedule communication channels. However, multi-radio scheduling over a limited set of channels to minimize the effect of interference and maximize network performance in the presence of concurrent network flows remains a challenging problem. The state-of-the-practice in channel scheduling of multi-radios reveals not only gaps in achieving network capacity but also significant communication overhead. 

This paper proposes an adaptive channel hopping algorithm for multi-radio communication, QuickFire MAC (QF-MAC), that assigns per-node, per-flow ``local'' channel hopping sequences, using only one-hop neighborhood coordination. QF-MAC achieves a substantial enhancement of throughput and latency with low control overhead. QF-MAC also achieves robustness against network dynamics, i.e., mobility and external interference, and selective jamming attacker where a global channel hopping sequence (e.g., TSCH) fails to sustain the communication performance. Our simulation results quantify the performance gains of QF-MAC in terms of goodput, latency, reliability, communication overhead, and jamming tolerance, both in the presence and absence of mobility, across diverse configurations of network densities, sizes, and concurrent flows.
\end{abstract}


\section{Introduction}
Wireless mesh networks enable a cost-effective option for IoT connectivity in resource-, infrastructure- and access-limited settings, wherein device-to-device communications are achieved by multi-hop cooperation. The growing throughput-latency demands of 5G applications and the considerations of mobility and interference have motivated the development of performance-efficient networking platforms that can adapt to dynamic network environments. To that end, multi-channel multi-radio (MCMR) offers one such platform for high throughput-latency performance by increasing the degree of concurrent communication. This paper focuses on MCMR medium access control (MAC) protocols, which themselves play a key role in maximizing the performance in the presence of concurrent communication by scheduling transmissions in localities to avoid interference.


MCMR MACs have leveraged time diversity typically, i.e., using carrier sense multiple access with collision avoidance (CSMA/CA) to delay transmissions when needed or randomized transmission. In addition,  frequency diversity, i.e., channel hopping or orthogonal frequency-division multiplexing, is applied to deal with interference. However, the state-of-the-practice channel scheduling of MCMR networks has several gaps in achieving high throughput-latency performance. This is especially true when there are a number of concurrent traffic flows or when low latency is desired for bursty traffic.

As a representative example, let us consider the IEEE 802.15.4e Time Slotted Channel Hopping (TSCH) \cite{IEEE802.15.4} protocol. TSCH has found significant adoption in industrial low-power applications of wireless sensor-actuator networks (WSAN) and multiple industrial standards \cite{WirelessHART, ISA100, dujovne20146tisch}. As the demand for high data-rate services has grown, several TSCH-based protocols have been proposed for MCMR networks \cite{banik2017improving, baddeley20206tisch++}. 

The IEEE 802.15.4 standard leaves the choice of the channel sequence scheduler up to the protocol implementer. Scheduling may be realized in a centralized \cite{palattella2012traffic, soua2012modesa} or distributed \cite{accettura2013decentralized, soua2016wave} manner. However, channel hopping schedulers for TSCH typically rely on a globally fixed channel sequence. The sequence they use is relatively long, as the use of short sequences tends to be insufficient to mitigate interference, especially in the presence of concurrent traffic flows. The reliance on a large number of narrowband channels, in turn, yields a tradeoff with high spectrum utilization. Moreover, the fixed channel schedule underperforms in the presence of network dynamics. TSCH does not handle mobile scenarios effectively \cite{orfanidis2021tsch}, as the channel sequence does not adapt to the changes in network interference with the change in network topology. TSCH is also vulnerable to jamming attacks \cite{cheng2019cracking, cheng2021launching}, as an attacker can reverse engineer the channel hopping sequence of a radio through observation of channel activities and then launch selective jamming attacks against all outgoing packets on the hopping sequence channels.

\noindent
\textbf{Contributions of the paper.} In this paper, we design a control overhead efficient MAC scheduler that overcomes the limitations of TSCH but also other (e.g., CSMA/CA-based) schedulers under MCMR networks. Our protocol, which we call QuickFire MAC (or QF-MAC for short), uses a locally adaptive channel hopping sequence assignment to achieve reliable wireless communication with high throughput-latency performance in MCMR meshes, as follows.  

The key idea in QF-MAC is to use a per-node, per-traffic flow ``local'' channel hopping sequence. Only a short local sequence is needed, which yields improved throughput efficiency across a variety of traffic loads. The selection of the local sequence for a new flow at a node takes into account the sequences at potentially interfering nodes in the same flow and those of other flows going through that node. Notably, QF-MAC assigns its channel hopping sequences using only one-hop neighbor coordination. The net result is that each node efficiently avoids the end-to-end intra-flow interference from source to destination for the new flow and also its self-interference from the intersecting flows assigned to its other radios.

By itself, the assignment of channel sequences with local coordination is insufficient for dealing with network dynamics---whether they are due to adversarial jamming, external interference, internal interference across larger network regions, or node mobility. The second key idea of QF-MAC is to deal with all these dynamics by adapting the channel hopping sequence at each node to use only goodput-efficient channels. This adaptation also has low control overhead by leveraging passive interference estimation at each node. 

These two ideas yield a resilient MAC protocol compatible with different routing protocols and diverse network environments. QF-MAC makes few assumptions, i.e., about the number of channels, network density, node mobility, or the predictability of external or internal interference. 

We emulate real code implementations over simulated networks to validate the high performance of QF-MAC with a limited set of channels. We quantify its high goodput efficiency, low latency, low control overhead, and high packet reception ratio. Our validations span environments with and without adversarial jammers, with and without mobility, and a range of concurrent flows, network sizes, and densities. We also comparatively evaluate its improved jamming and mobility tolerance with respect to TSCH and CSMA/CA protocols in MCMR wireless meshes. 

The rest of the paper is organized as follows. In Section \Romannum{2}, we discuss related work, including TSCH. In Sections \Romannum{3} and \Romannum{4}, we formalize the scheduling problem and then present our solution for scheduling local channel hopping sequences and realizing channel adaptation with respect to internal and external interference. We present evaluation results in Section \Romannum{5}, and make concluding remarks in Section \Romannum{6}.

\section{Related Work}
\subsection{Time Slotted Channel Hopping}
TSCH is a channel hopping MAC protocol supported by IEEE 802.15.4e \cite{IEEE802.15.4}. It inherits slotted time access from the IEEE 802.15.4 standard. To achieve reliable performance in the presence of interference and multipath fading, it leverages multi-channel communication and channel hopping based on a synchronized slotframe. Each slotframe is a collection of slots with a fixed length that repeats over time; every slot is long enough to transmit a data packet and receive an acknowledgment between a pair of nodes within some transmission radius.  

Channel hopping in TSCH relies on channel offset schedules so as to avoid internal interference (and, in turn, communication collisions). The offset schedule is over a predefined sequence of channels, i.e., a globally shared sequence, instead of a randomly generated channel sequence. This is done primarily to reduce the channel synchronization overhead. The communication channel to be used for transmission and reception over a link at a slot follows the function $CHS$:
\begin{equation}
\begin{array}{l}
Channel = CHS[(ASN + OF) \ mod \ L],
\label{CHS}
\end{array}
\end{equation}
where $ASN$ stands for the absolute slot number, $OF$ is the channel offset of a communication link between two nodes, and $L$ denotes the length of a channel hopping sequence. The idea is to assign a different $OF$ to links that can potentially interfere.  
Several heuristics \cite{hermeto2017scheduling} have been proposed to maximize the end-to-end reliability and minimize the end-to-end delay in TSCH MAC protocols, either by centralized or by distributed scheduling of a node's transmission time slots (time diversity) and the channel offset in Eq.~\ref{CHS} (frequency diversity). 

Even though TSCH-based protocols have shown good performance in industrial environments with stable topologies and low traffic, the reliance on a globally shared channel sequence limits their robustness to network dynamics. When $L$ becomes small, its ability to avoid interference is negatively impacted. Furthermore, TSCH's channel hopping is vulnerable to a selective jamming attacker, given its lack of channel adaption. Also, its use of a global channel hopping sequence inherently limits its resilience with respect to concurrent communications. The number of transmission links within an interference region must be at most $L$ to avoid collision\footnote{We assume there is no repeated channel in a channel hopping sequence.}. If $L$\!+\!1 links exist within the interference region, there exist at least two links that will collide predictably. In this case, the achievable capacity is reduced by at least $\frac{1}{L}$, which is more than the expected reduction when scheduling with random channel sequences.

\subsection{Channel Hopping in MCMR Networks}
Motivated by the spectrum efficiency gain they offer, MCMR mesh networks are increasingly being adopted in 5G contexts. The gain results primarily from the support of concurrent communications over orthogonal channels in MCMR networks. This support requires achieving rendezvous among communicating devices so that they are switched to the same channel at the same time. As illustrated by TSCH, channel hopping sequences provide a convenient basis for programming rendezvous synchronization, and so several channel hopping protocols for MCMR platforms have been proposed in the literature \cite{lin2013channel, yang2015fully, tan2019difference, lin2019homogeneous, chao2020adjustable}. However, most extant work in MCMR MAC channel hopping schedulers only considers internal interference avoidance under static networks. Our solution further provides reliability and survivability to the MCMR networks with dynamics of external interference, adversarial jamming, and node mobility.

Channel hopping sequences have also been exploited to improve the reliability of the network, essentially by amortizing the impact of channels with high interference or multipath fading over the set of channels. The authors of \cite{banik2017improving} accordingly use channel hopping sequences to extend the centralized Traffic Aware Scheduling Algorithm (TASA) \cite{palattella2013optimal}, to enable concurrent transmissions with interference avoidance in a slot-synchronized, IEEE 802.15.4, multi-hop network setting. In addition, a PHY design is proposed to support concurrent transmissions with IPv6 traffic within IEEE 802.15.4 networks \cite{baddeley20206tisch++}. 


Channel hopping has moreover been considered from a security perspective in MCMR meshes, more specifically, to enable jamming tolerance. The authors of \cite{khattab2008modeling} maximize network goodput in the presence of jamming attacks by combining channel-hopping with error-correcting codes (ECC) to lower the blocking probability over channels. Furthermore, in \cite{khattab2008jamming} they compare reactive and proactive channel hopping schemes and show that the reactive approach provides better anti-jamming performance than the proactive one in MCMR networks with respect to communication availability.

\section{Problem Statement}
The primary objective of MACs in MCMR mesh networks is to optimize the communication performance of nodes by minimizing the effect of interference within each node's interference region. Doing so impacts various performance metrics of multi-hop flow communications beyond throughput efficiency, reception ratio, latency, and control overhead.  

In what follows, we assume that each traffic flow corresponds to a simple path between a source and destination node comprising alternating network node radios and links. Without loss of generality, we let each flow be pinned to a specific radio at the nodes it traverses. (Our problem statement is readily extended to formulations where flows consist of a network of paths.) 


In order to characterize the interference that results from a transmitting radio within its interference region and the efficiency of a channel hopping scheduler, we define a channel efficiency measure. Let $p_{i, ch, t}$ be the probability that radio $i$ at a node sends out a packet over channel $ch$ at time slot $t$ and $\Phi(i)$ denotes the number of internal and external interfering sources within radio $i$'s interference region. (Note that other radios in the same node as $i$ are among the sources that are within its interference region.) Given sender $i$ and its interference set $\Phi(i)$, the channel efficiency of channel $ch$ within $\Phi(i)$ at time slot $t$ is defined as follows:

\vspace*{-3mm}
\begin{equation}
\begin{array}{l}
e_{\Phi(i), ch, t} = \sum_{j\in \Phi(i)}{p_{j, ch, t} \prod_{k\in \Phi(i) \setminus \{j\}}{(1-p_{k, ch, t})}}
\label{CH_UTIL}
\end{array}
\vspace*{1mm}
\end{equation}

Let $f(i, t)$ be a channel hopping scheduling function that specifies the channel used at radio $i$ at slot $t$. Accordingly, the average scheduling efficiency, $\hat{E}$, may be calculated as follows:

\vspace*{-1mm}
\begin{equation}
\begin{array}{l}
\hat{E}(f)= \frac{1}{T} \sum_{t=0}^{T-1} {\sum_{i \in V}{e_{\Phi(i), f(i, t), t}}}
\label{AVG_UTIL}
\end{array}
\vspace*{1mm}
\end{equation}
where $T$ is the total number of slots, and $V$ is the set of radios in the network. 
{\em The channel scheduling problem is to determine a scheduler function $f$, such that $\hat{E}$ is maximized to minimize the interference}.

All notations related to the channel scheduling problem, and our solution QF-MAC, are listed in Table~\ref{notation}.

\begin{table}[t]
\caption{Notations for QF-MAC}
\begin{center}
\begin{tabular}{c|c}
\hline
\textbf{Symbol}&\textbf{Meaning} \\
\hline
$U$ & set of all channels\\
\hline
$\Delta$ & interference-to-reliable-transmission ratio\\
\hline
$C$ & number of radios at a node\\
\hline
$L$ & length of a channel hopping sequence\\
\hline
$f_x$ & flow id\\
\hline
$i\in V$ & radio id\\
\hline
$(i, j)\in E$ & link id\\
\hline
$ch\in U$ & channel id\\
\hline
$t$ & time slot number\\
\hline
$\Phi(i)$ & set of radios within an interference region\\
\hline
$p_{i,ch,t}$ &radio's probability of outgoing packet\\
\hline
$e_{\Phi(i),ch,t}$ &channel efficiency in an interference set\\
\hline
$CHS_{Tx/Rx}(f_x, (i, j))$ & Tx/Rx channel sequence over a flow link\\
\hline
$S_{\Phi(i)}$ & set of channels in all active Tx channel\\
 & sequences in an interference set\\
\hline
$ag_{\Phi(i),ch,t}$ &aggregate channel traffic in an interference set\\
\hline
$G(i,ch)$ & radio's goodput efficiency of over a channel \\
\hline
\end{tabular}
\label{notation}
\end{center}
\end{table}

\section{QF-MAC Protocol}

Towards solving the optimization problem in Eq.~\ref{AVG_UTIL}, our design of QF-MAC begins by assigning, for each flow in the network, a channel sequence for each node associated with that flow (cf.~Section IV.B). Since each flow at a node is pinned to a unique radio, we can equivalently regard its corresponding channel sequence as being associated with the pinned radio. The channel sequence is chosen in coordination with the 1-hop neighbors in the flow, such that intra-flow interference is avoided on an end-to-end basis. Therefore, it allows for concurrent communication along all radios of the same flow (cf.~Section IV.B.I). It is moreover chosen such that it deterministically avoids conflicts with the other channel sequences assigned to the other radios in that node (cf.~Fig.~\ref{CHS} as well as Section IV.B.II).

\begin{figure*}[hbt!]
    \centering
     \subfigure
     {
        \includegraphics[width=0.88\textwidth]{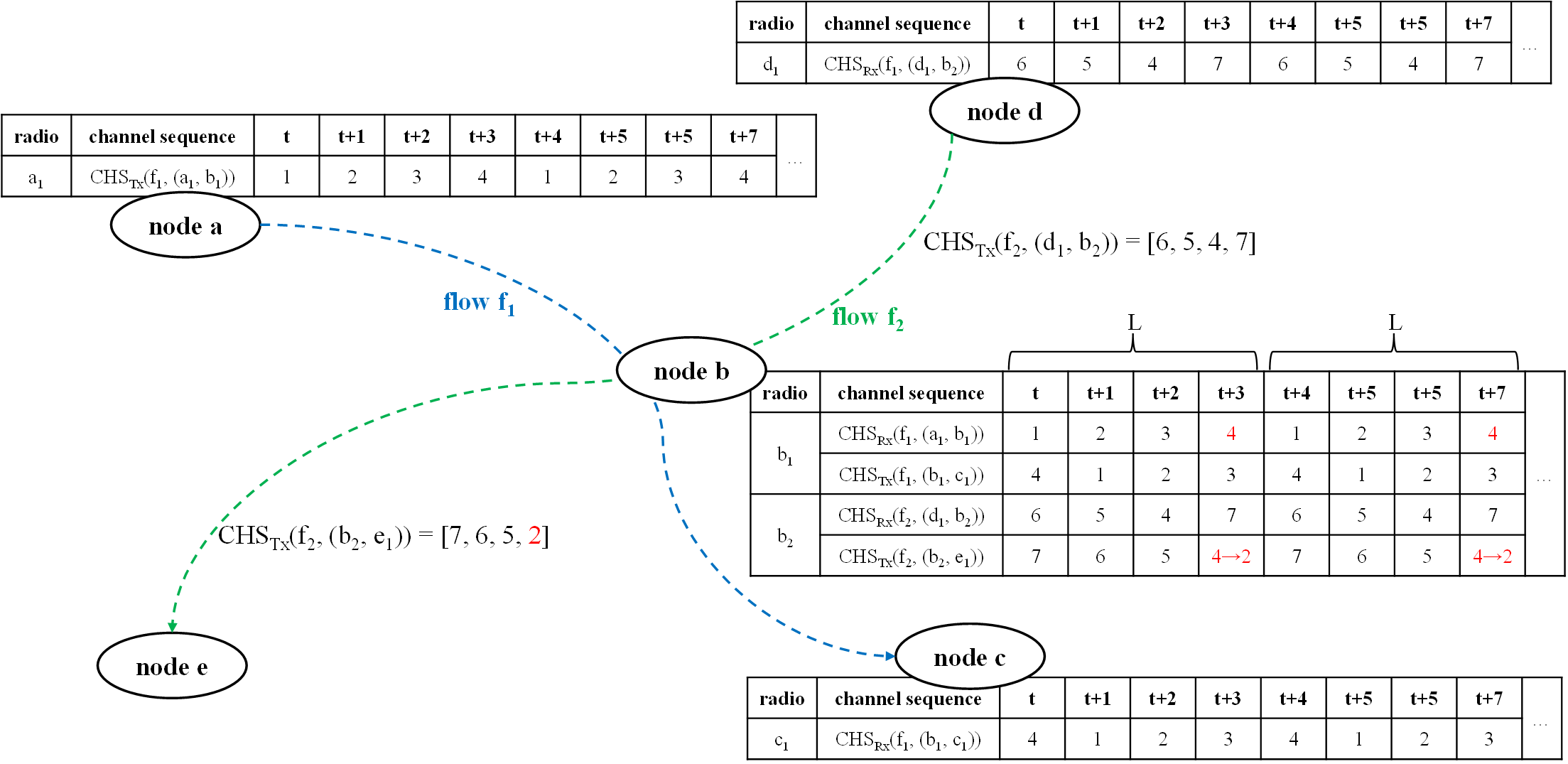}
    }
    \caption{Local channel sequences for two flows in a network. The example illustrates how QF-MAC avoids intra-flow and intersecting-flow interference locally. Assume flow $f_1$ is initiated at radio $a_1$ with route $(a_1, b_1, c_1)$ to assign Tx and Rx channel sequences in radio $a_1$, $b_1$, and $c_1$. Later flow $f_2$ is initiated at radio $d_1$ with route $(d_1, b_2, e_1)$ to send a Tx sequence, $CHS_{Tx}(f_2, (d_1, b_2))$, to radio $b_2$. Radio $b_2$ identifies the intersecting-flow conflict between $CHS_{Rx}(f_1, (a_1, b_1))$ and $CHS_{Tx}(f_2, (b_2, e_1))$ and then replaces $ch4$ in $CHS_{Tx}(f_2, (b_2, e_1))$ with a non-conflicted channel, $ch2$, which also avoids intra-flow interference in both $f_1$ and $f_2$. }
    \label{CHS}
\vspace*{3mm}
\end{figure*}

The design of QF-MAC also indirectly deals with internal interference with other nodes, even potentially malicious external interference effects, by adapting the flow sequences as follows. Each node locally updates each of its channel sequences in the presence of varying interference by changing goodput-inefficient channels such that the scheduling efficiency, $\hat{E}$, is optimized over time for the network (cf.~Section IV.C). Goodput-inefficient channels are identified by leveraging a passive local interference estimation mechanism. Three forms of channel sequence update are performed to improve goodput efficiency--channel augmentation (by adding a channel), channel migration (by substituting a channel), and channel reduction (by dropping a channel).



\subsection{Choosing the Sequence Length for Interference Avoidance}
Two factors contribute to characterizing the minimum length of any channel hopping sequence. First, since nodes with $C$ radios can support up to $C$ concurrent transmissions at a node, to avoid self-interference at nodes where concurrent flows intersect, which we refer to as ``intersecting-flow'' interference, the sequence length $L$ must be at least $C$. The second factor results from our design wherein we seek to pipeline flows to maximize concurrent transmissions and, thereby, throughput efficiency. To that end, we assume a standard unit-disk interference model \cite{von2005robust}. We let $\Delta$ denote the ratio of interference distance to the reliable transmission distance, and $R_{i}$ and $R_{t}$ denote the interference radius and reliable transmission radius; $\Delta = \lceil \frac{R_{i}}{R_{t}} \rceil$. It follows that to allow concurrent communication over $2\Delta\!+\!1$ channels while avoiding intra-flow interference within the interference region, the number of distinct channels in $L$ should be at least $2\Delta\!+\!1$.    

Hence, the minimum channel hopping sequence length to avoid interference is stated in Eq.~\ref{SEQ_LEN}, where $U$ is the set of all channels. 

\vspace*{-1mm}
\begin{equation}
\begin{array}{l}
|U| \geq L \geq max(2\Delta+1, C)
\label{SEQ_LEN}
\end{array}
\vspace*{1mm}
\end{equation}
A minimum value of $L$ is then applied to all radios to minimize the control and computation overhead. 

\subsection{Local Assignment of Channel Hopping Sequences}
In an accompanying theoretical analysis, relegated to Appendix~\ref{FirstAppendix}, that considers concurrent flows within the interference region of any node, we show that the ratio of collisions/slot of per-flow sequence schedulers (i.e., QF-MAC) is better than that for global sequence schedulers (i.e., TSCH), given $L \leq |U|$. Based on this insight, we design a scheme to assign per-flow channel sequences that notably require only one-hop neighbor coordination.  

Assume that a flow has been admitted to the network by choosing its potentially multi-hop routing path and assigning a radio in each of the path nodes to the flow, yielding a route of say  $(v_{1}, ..., v_{d})$ with $d\!-\!1$ links. Starting with radio $v_1$, QF-MAC calculates and assigns local channel sequence for the radios in the chosen route; for radio $v_r$, with $1 < r \leq d$, QF-MAC does so with Algorithm~\ref{CHANNEL_SEQ_ASSIGN}. Its channel sequence assignment on $v_{r}$ minimizes the interference both from all radios associated with the same flow within $\Phi(v_{r})$ (intra-flow interference) and from all radios at a node (intersecting-flow interference). QF-MAC then sends a channel sequence request to the successor $v_{r+1}$ over the control channel and subsequently starts communication over the link $(v_r, v_{r+1})$ using the channel sequence of $v_r$. All control messages in QF-MAC are sent via a predefined control channel so that no additional scheduling overhead is involved. 



\noindent
\textbf{I. Intra-flow Interference Avoidance.} 
When a flow $f_x$ is initiated given a route $(v_1, ..., v_r, ..., v_d)$, the source radio $v_1$ first randomly selects a channel sequence $CHS_{Tx}(f_x, (v_1, v_2))$ consisting of $L$ non-repeating random channels for link $(v_1, v_2)$ and then sends the channel sequence to the next hop, $v_2$. This is shown in Lines 3 to 5 of Algorithm~\ref{CHANNEL_SEQ_ASSIGN}. Upon receiving the channel sequence from $v_{r-1}$, as shown in Line 7, $v_r$ sets up the scheduling for data reception over $(v_{r-1}, v_r)$. If $v_r$ has an outgoing link for $f_x$, as shown in Lines 8 to 9, it then adopts by default the channel sequence for the transmissions on link $(v_r, v_{r+1})$ right shifted by one position, according to Algorithm~\ref{RIGHT_SHIFT}, so as to avoid intra-flow interference between the two links $(v_{r-1}, v_r)$ and $(v_r, v_{r+1})$. Along the flow path, the intra-flow interference is thus successively avoided by repeating this assignment until the destination $v_d$ is reached.\footnote{For paths with a round shape, intra-flow interference may result if $v_r$ and $v_{r+L}$ are within $\Phi(v_r)$ and $(OF_{v_r} \! - \! OF_{v_{r+L}}) \ mod \ L = 0$. We handle this case implicitly via the inter-flow interference handling Channel Adaptation mechanism, described in Section~IV.C.}

As will be described in Section~IV.B.II, a radio $v_r$ may need to change its channel sequence from the default one to accommodate channel sequence conflicts from the other flows through its node. 
To avoid intra-flow interference in this case as well, the channel sequence request sent to $v_{r+1}$ includes not just the channel sequence of $v_r$ but also the sequence information of $v_r's$ $\Delta\!\!-\!\!1$-hop ancestors\footnote{QF-MAC allows each radio to change its Tx channel sequence locally to avoid intersecting-flow interference. The update is nontrivial since $v_r's$ $\Delta\!\!-\!\!1$-hop ancestors can use different channel sequences.}. Accordingly, $v_{r+1}$ avoids intra-flow interference using the potentially different channel sequences being used within its interference region in the flow. We observe in passing that a feasible assignment for avoiding intra-flow and intersecting-flow interference exists as long as the assumption that $L$ satisfies Eq.~\ref{SEQ_LEN} holds.      

\begin{algorithm}[t]
Input: flow id $f_x$, incoming flow link $(v_{r-1}, v_r)$ with channel sequence $CHS_{Tx}(f_x, (v_{r-1} ,v_r))$,  outgoing flow link $(v_r, v_{r+1})$ \\
Output: channel sequence $CHS_{Tx}(f_x, (v_r, v_{r+1}))$\\
\SetAlgoLined
    \eIf{$v_r$ = src($f_x$)}{
       $CHS_{Tx}(f_x, (v_r, v_{r+1}))$ :=  $L$ non-repeating random channels selected from  $U$\;  
	schedule send of $CHS_{Tx}(f_x, (v_r, v_{r+1}))$ to $v_{r+1}$\;
  }{
    $CHS_{Rx}(f_x, (v_{r-1} ,v_r))$ :=  $CHS_{Tx}(f_x, (v_{r-1} ,v_r))$\;
    \If{$v_r$ != dest($f_x$)}{
      $CHS_{Tx}(f_x, (v_r, v_{r+1}))$ :=  $RIGHT{-}SHIFT(CHS_{Tx}(f_x, (v_{r-1} ,v_r)), 1)$\;
    }

    \eIf{($\exists \ l, f_y$ s.t. $CHS_{Rx}(f_x, (v_{r-1} ,v_r))[l]$ = $CHS_{Rx}(f_y, (v_q, v_r))[l]$)}{
      schedule send of reject of $CHS_{Tx}(f_x, (v_{r-1} ,v_r))$ to $v_{r-1}$\;
    }{
      \If{$v_r$ $\neq$ dest($f_x$) $\wedge$ ($\exists \ l, f_y$ s.t. $CHS_{Tx}(f_x, (v_r, v_{r+1}))[l]$ = $CHS_{Rx}(f_y, (v_q, v_r))[l]$)}{
        replace $CHS_{Tx}(f_x, (v_r, v_{r+1}))[l]$ with a channel without conflict\;
      }
      \If{$v_r$ $\neq$ dest($f_x$) $\wedge$ ($\exists \ l, f_y$ s.t. $CHS_{Tx}(f_x, (v_r, v_{r+1}))[l]$ = $CHS_{Tx}(f_y, (v_r, v_s))[l]$)}{
        replace $CHS_{Tx}(f_x, (v_r, v_{r+1}))[l]$ with a channel without conflict\;
      }
      \If{($\exists \ l, f_y$ s.t. $CHS_{Rx}(f_x, (v_{r-1} ,v_r))[l]$ = $CHS_{Tx}(f_y, v_r, v_s)[l]$)}{
        replace $CHS_{Tx}(f_y, (v_r, v_s))[l]$ with a channel without conflict\;
        schedule send of $CHS_{Tx}(f_y, (v_r, v_s))$ to $v_s$\;
      }
    }
    
    schedule to send $CHS_{Tx}(f_x, (v_r, v_{r+1}))$ to $v_{r+1}$\;
  }
\caption{CHANNEL-SEQ-ASSIGN at radio $v_r$}
\label{CHANNEL_SEQ_ASSIGN}
\end{algorithm}

\begin{algorithm}
\SetAlgoLined
  \For{$l = 0 .. L-1$}{
     $CHS'[(l+of\!fset) \% L]$ := $CHS[l]$\;
  }
   \Return $CHS'$\;
\caption{RIGHT-SHIFT(CHS, offset)}
\label{RIGHT_SHIFT}
\end{algorithm}

\noindent 
\textbf{II. Intersecting-flow Interference Avoidance.} 
Since each node is equipped with $C$ full-duplex radios, it can perform up to $C$ concurrent Tx+Rx communications. Each radio $v_r$ associated with flow $f_x$, has its incoming Rx channel sequence, $CHS_{Rx}(f_x, (v_{r-1} ,v_r))$, and its outgoing Tx channel sequence, $CHS_{Tx}(f_x, (v_r, v_{r+1}))$.  To avoid interference with the flows at other active radios at the same node, QF-MAC checks if any channel in $CHS_{Rx}(f_x, (v_{r-1},v_r))$ or  $CHS_{Tx}(f_x, (v_r, v_{r+1}))$ conflicts with some other radio's Tx/Rx channel sequences. 
In order to save control overhead, QF-MAC allows $v_r$ to locally change only its Tx channel sequence of the link $(v_r, v_{r+1})$ to avoid interference; the net result is that the channel sequence scheduling of a path continues to be completed in a downstream direction. 

The change is accomplished as follows: As shown in Lines 11--24 in Algorithm~\ref{CHANNEL_SEQ_ASSIGN}, four types of conflicts are considered: $(Rx_{f_x}, Rx_{f_y})$, $(Rx_{f_x}, Tx_{f_y})$, $(Tx_{f_x}, Rx_{f_y})$, and $(Tx_{f_x}, Tx_{f_y})$ conflicts, wherein the left and right terms of a tuple stand respectively for the newly scheduled sequences of flow $f_x$ and the existing sequence from other radios associated with flow $f_y$. 
Lines 11--12 specify the handling of the $(Rx_{f_x}, Rx_{f_y})$ conflict case: Since $v_r$ is not allowed to modify the Rx channel sequence of $(v_{r-1}, v_r)$, it sends a reject with its channel sequence scheduling information to $v_{r-1}$ so that $v_{r-1}$ can reschedule a channel sequence for link $(v_{r-1}, v_r)$. 
In the $(Tx_{f_x}, Rx_{f_y})$ conflict case, as shown in Lines 14--15, and the $(Tx_{f_x}, Tx_{f_y})$ conflict case, as shown in Lines 17--18, the conflicted channel in $CHS_{Tx}(f_x, (v_r, v_{r+1}))$ is replaced directly to exclude the collision. Note that the replacement should exclude all the conflicted channels used in both the other radios at the node itself and the $\Delta\!\!-\!\!1$-hop ancestors of flow $f_x$. As shown in Lines 20--22, the case of $(Rx_{f_x}, Tx_{f_y})$ is resolved by changing the Tx channel sequence of $f_y$ and then sending the updated channel sequence request to the corresponding next hop, $v_s$. 

Note that an update of $CHS_{Tx}(f_y, (v_r, v_s))$ may cause intra-flow interference at the radios associated with $f_y$. This case is handled implicitly via the inter-flow interference handling in Channel Adaptation (cf.~Section~IV.C).

\subsection{Channel Adaptation}
Unlike intra-flow and intersecting-flow interference, the variation of inter-flow and external interference across node locations and time is not controlled by QF-MAC. QF-MAC does, however, adapt to these interference dynamics to optimize channel efficiency over time. We recall that Eq.~\ref{CH_UTIL} states that channel efficiency $e_{\Phi(i), ch, t}$ is optimized when $\sum_{v\in \Phi(i)}{p_{v, ch, t}}=1$\footnote{Namely, in channel hopping schedulers, only one radio transmits a packet over $ch$ at $t$ within $\Phi(i)$.}. Moreover, we recall that \cite{li2010chameleon} shows that $\sum_{v\in \Phi(i)}{p_{v, ch, t}}$ can be approximated by a local estimation of $p_{i, ch, t}+e^{I(ch, i, t)}$,  where $I(ch, i, t)$ is an interference estimator of the probability that some interferers $j\in \Phi(i) \setminus \{i\}$ transmit on channel $ch$ at slot $t$. And furthermore, IEEE 802.11 and IEEE 802.15.4 utilize the evaluation of clear channel access (CCA) in the physical layer to detect the level of interference, $I(ch, i, t)$. All of these observations make it feasible in QF-MAC to have each radio $i$ locally and efficiently estimate channel efficiency $e_{\Phi(i), ch, t}$, by measuring incoming traffic, $p_{i, ch, t}$, and the interfering traffic on channel $ch$. Notably, this estimation is achieved without introducing any extra communication overhead.

By adopting the use of packet acknowledgments for transmissions, QF-MAC also calculates the goodput efficiency, $G(i, ch)$, in terms of the ratio of ACKed packets to transmitted packets at radio $i$ on channel $ch$, and thereby characterizes the effect of external interference. In turn, it uses the estimates of the channel efficiency and the goodput efficiency to mitigate inter-flow and external interference, as follows.

\vspace*{2mm}
\noindent
\textbf{Inter-flow and External Interference Avoidance.} Let $S_{\Phi(i)}$ be the set of channels used in all active Tx channel sequences from radios in $\Phi(i)$, $|U| \! \geq \! |S_{\Phi(i)}| \geq L$\footnote{$S_{\Phi(i)}$ can be estimated through monitoring CCA signals over all channels at radio $i$.}. Also, let $ag_{i,ch,t} = \sum_{j\in \Phi(i)}{p_{j, ch, t}}$ denote the aggregate traffic in channel $ch$ in the interference set of radio $i$ at slot $t$,  It follows that $\sum_{ch \in S}{ag_{\Phi(i),ch,t}} \! > \! |S_{\Phi(i)}|$ implies that in $\Phi(i)$ the aggregate traffic load from all channels in $S_{\Phi(i)}$ exceeds the achievable capacity. QF-MAC improves the transmission performance at $i$ by introducing a channel out of $S_{\Phi(i)}$ to replace the one with the largest $ag_{\Phi(i), ch, t}$ in all active Tx channel sequences at the node. On the other hand, $\sum_{ch \in S}{ag_{\Phi(i), ch, t}} \leq |S_{\Phi(i)}|$ and a poor value of $G(i, ch)$ reveal an increased impact from external interference on a particular  channel. QF-MAC selects a new channel for radio $i$ from $S_{\Phi(i)}$ to replace the channel with the lowest $G(i, ch)$ among all of its Tx channel sequences. Algorithm~\ref{CHANNEL_ADAPT} specifies the rules of channel adaptation with respect to three cases: channel augmentation, channel migration, and channel reduction.


\begin{algorithm}[t]
\SetAlgoLined
    \eIf{$\sum_{ch \in S}{ag_{\Phi(i),ch,t}} > |S_{\Phi(i)}|$}{
      $ ch_{old} = \operatorname*{argmax}_{ch \in S_{\Phi(i)}} ag_{\Phi(i),ch,t}$\;  
      $X:= U \setminus S_{\Phi(i)}$\;
      select a random $ch_{new} \in X$\;
      replace $ch_{old}$ by $ch_{new}$ in $\{CHS_{Tx}(i, j) \ | \  ch_{old} \in CHS_{Tx}(i, j) \}$\;
	send updated $CHS_{Tx}(i, j)$ to  $j$\;
  }{
    \uIf{$(\exists \ ch \in S_{\Phi(i)}: G(i, ch) < \theta_1)$}{
      $ ch_{old} = \operatorname*{argmin}_{ch \in S_{\Phi(i)}} G(i, ch)$\;  
      $X:= $ all non-conflicting channels in $S_{\Phi(i)}$\;
    }
     \uElseIf{$\sum_{ch \in S_{\Phi(i)}}{ag_{\Phi(i),ch,t}} < \theta_2|S_{\Phi(i)}|$}{
      $ ch_{old} = \operatorname*{argmin}_{ch \in S_{\Phi(i)}} ag_{\Phi(i),ch,t}$\;  
      $X:= $ all non-conflicting channels in $S_{\Phi(i)}$\;
      
    }
    select $ch_{new} \in X$ with probability  $\frac{ag_{\Phi(i),ch_{new},t}^{-1}}{\sum_{ch \in X}{ag_{\Phi(i),ch,t}^{-1}}}$\;
    replace $ch_{old}$ by $ch_{new}$ in $\{CHS_{Tx}(i, j) \ | \ ch_{old} \in CHS_{Tx}(i, j) \}$\;
    send updated $CHS_{Tx}(i, j)$ to $j$\;
  }
\caption{CHANNEL-ADAPT($i$)}
\label{CHANNEL_ADAPT}
\end{algorithm}

{\em Channel augmentation}, as shown in Lines 1--6, is triggered whenever the aggregate traffic load over all channels in ${S_{\Phi(i)}}$ is larger than the achievable capacity. In Line 2, radio $i$ selects for replacement a channel $ch_{old}$ whose traffic within its interference region is the largest.  Note that $ag_{\Phi(i),ch_{old},t}$ is always greater than 1, which implies a high level of inter-flow interference and hence the likelihood that the substitution of $ch_{old}$ at radio $i$ will improve the performance at radio $i$. In Lines 3--4, a new channel from $U\setminus S_{\Phi(i)}$ is introduced to increase the number of channels used for concurrent transmission in $\Phi(i)$ to reduce the internal interference. 
In Lines 5--6, the newly selected channel replaces $ch_{old}$ in all Tx channel sequences used in radio $i$, and then it notifies the corresponding neighbor about the update of the channel sequence.

{\em Channel migration}, as shown in Lines 8--10, swaps out the most goodput-inefficient channel from Tx sequences. Note that Line 8 captures the case with short-term interference dynamics--i.e., $ag_{\Phi(i), ch, t} \leq |S_{\Phi(i)}|$--but where at least one channel has goodput below the threshold of $\theta_1$. In Lines 9--10, radio $i$ replaces the channel, $ch_{old}$, with minimal goodput efficiency with a new channel from $S_{\Phi(i)}$, selected in Line 14, that does not conflict any other channels in radio $i$'s Tx and Rx channel sequences at $t$. 
The probability of selecting a channel is inversely proportional to its aggregate traffic. Instead of deterministic selection, we adopt a probabilistic rule that eschews the situation where nodes with similar estimated states will simultaneously swap to the same channel. Then the update of Tx channel sequences, in Lines 15--16, is sent to the receiver $j$ of link $(i, j)$.


{\em Channel reduction}, as shown in Lines 11--13, shrinks the size of $S_{\Phi(i)}$ if the average channel efficiency is below a threshold $\theta_2$. Note that channel efficiency is maximized when $ag_{\Phi(i), ch, t} = 1$. Radio $i$ reduces the size of $S_{\Phi(i)}$ to be close to $L$ when the average data traffic within its interference region is low. This further reduces the overhead needed to accurately monitor channel states in $S_{\Phi(i)}$. 

\subsection{Compatibility with Routing Protocols}
QF-MAC works not only for path-oriented routing (i.e., on-demand routing) but also with table-driven routing protocols. In table-driven routing, the routing path of each packet may vary as the routing metrics change over time. When the path changes, QF-MAC deals with resulting changes as follows: If radio $i$'s next forwarder changes from $j$ to $k$ in its routing table, whenever the first packet arrives at $i$ after this change, radio $i$ sends a channel sequence request to $k$ and a channel sequence cancel to $j$, respectively, with only one-hop communication. Moreover, for the intra-flow interference from nodes in the flow at more than one hop distance, it effectively treats that as inter-flow interference: it spontaneously performs channel adaptation based on the estimation of the aggregate outgoing load and the interfering traffic in an interference region. 

\section{Performance Evaluation}
\subsection{Configuration Space of Simulated Networks}
Our validation of QF-MAC uses the ns3 simulator, which we extended in several ways: to support MCMR communication at each node, to support jammer nodes, and to emulate our real-code implementation of the MAC. (The extension to support emulation is based on an integration of ns-3 with the Direct Code Execution framework \cite{NS3DCE}). The validation compares the performance of QF-MAC with that of TSCH and CSMA/CA, in terms of goodput, end-to-end latency, control overhead, and packet reception ratio. It does so for environments with and without external interference and node mobility.  To calibrate the impact of leveraging local interference estimation and channel adaptation in tolerating adversarial jamming, we simulate two versions of QuickFire MAC: one with channel adaptation (QF-MAC-A) and the other without (QF-MAC). 

All MACs are tested with point-to-point flows whose routes are, in all cases, set up and maintained via a common reactive path-oriented routing protocol that finds a minimum interference path over the data plane for each flow. 


In the MCMR meshes we simulate, each node is equipped with four radios ($C=4$) that share a fixed capacity of 8 Mbps. For QF-MAC and TSCH, each radio operates in one of 8 channels, each with 1Mbps capacity. One channel is dedicated to sharing control messages (for routing and MAC); the remaining seven channels are used for data communication, i.e., $|U|=7$. For CSMA, each node operates with one dedicated channel of 1 Mbps for the control plane and two channels of 3.5 Mbps for data plane communication; when serving more than two flows, each radio uses round-robin scheduling for the flows. We note that we have used a modest number of channels and radios per node in part because that represents the state-of-the-art but also because it conveniently lets us study the ability of these MACs to handle interference. For the data plane, we apply a time-slotted MAC communication with a 10 ms slot for QF-MAC and TSCH and a 3 ms slot for CSMA. Every slot is long enough to transmit a 1000-byte data packet and receive an acknowledgment between a pair of nodes within some transmission radius.

The length of the channel sequence is set to its minimal value, 4, 
for QF-MAC and the maximal value (without repeated channel), $|U|=7$, for TSCH.  These assignments of $L$ satisfy Eq.~\ref{SEQ_LEN} while allowing for reduced control overhead for QF-MAC and the best opportunity for TSCH to avoid interference.   

Concerning interference, each radio is configured with the same transmission and interference radius of 1km ($\Delta=1$) followed by a unit-disk interference model \cite{von2005robust}. And with respect to simulation scenarios with external interference, a single adversarial jammer with a +3 dB above the nominal transmission power and a 100\% effective jamming range of 1.32km is placed in the center of the network layout. Outside of this radius, jamming contributes to node-local RF noise conditions and could still cause blocking interference depending on the SNR margin of the potential victim link, just as any other simulated RF emitter would. The jammer continually generates jamming signals to reduce the network data capacity by 2Mbps (i.e., from 7Mbps to 5Mbps) within its interference region, as follows:  For QF-MAC and TSCH, it jams the two channels with the highest aggregate traffic. 
For CSMA, it jams over a sub-band within one of the two 3.5 Mbps channels to reduce the available capacity by 2 Mbps.


The space of experiments we conducted is 4-dimensional, namely: (i) network size in $\{64, 125, 216\}$, (ii) network density\footnote{Density is defined as the average number of nodes per $R_t^2$ area, with $R_t$ denoting the transmission range} in $\{\sqrt[3]{3}, 2, 3, 4, 5\}$, (iii) number of concurrent 3 Mbit flows over in $\{1, 3, 7, 10\}$ , and (iv) mobility in $\{0, 10\}$ m/s followed by Gauss-Markov model. With respect to these four dimensions, we simulated networks with nodes distributed over a rectangular region using a uniform random distribution placement model. 
Each of our experiment configurations consists of four trials of 1 minute. All flows concurrently arrive after 5 seconds from the start of each trial. We note that each flow generates one packet at each slot. Each marked point in the lines in Figs.~\ref{fig_metrics_unjamming} and \ref{fig_metrics_jamming} represents an average result taken over the different numbers of flows. 


\begin{figure*}[hbt!]
    \centering
    \subfigure[Average Goodput]
    {
        \includegraphics[width=0.39\textwidth]{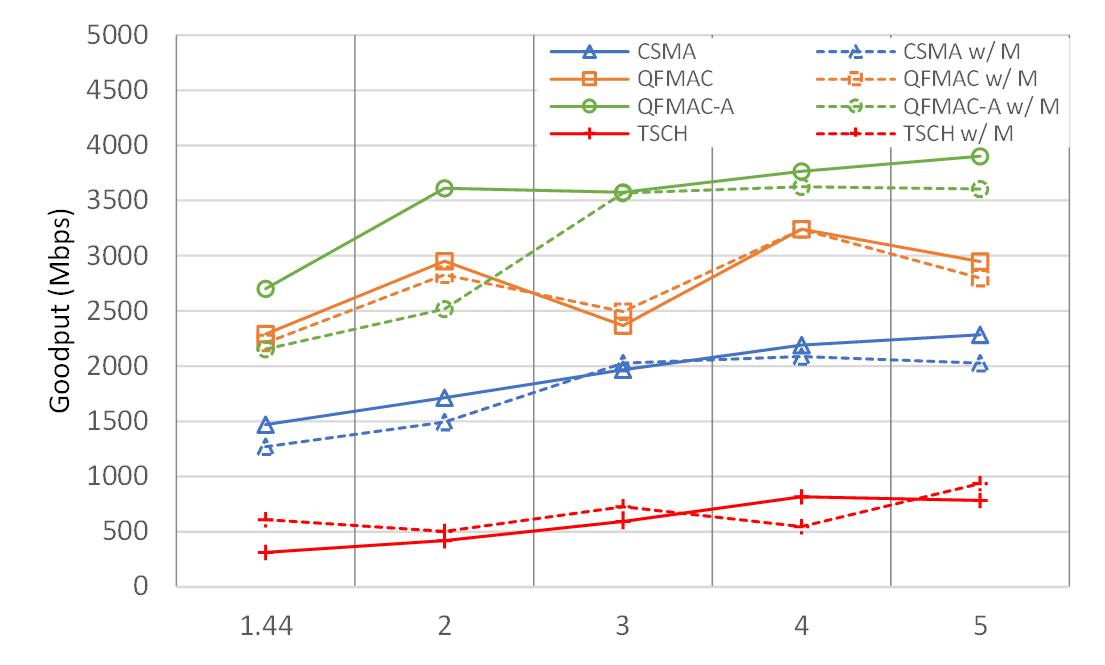}
        \label{fig_goodput_unjamming_size125}
    }
    \subfigure[Average Latency]
    {
        \includegraphics[width=0.39\textwidth]{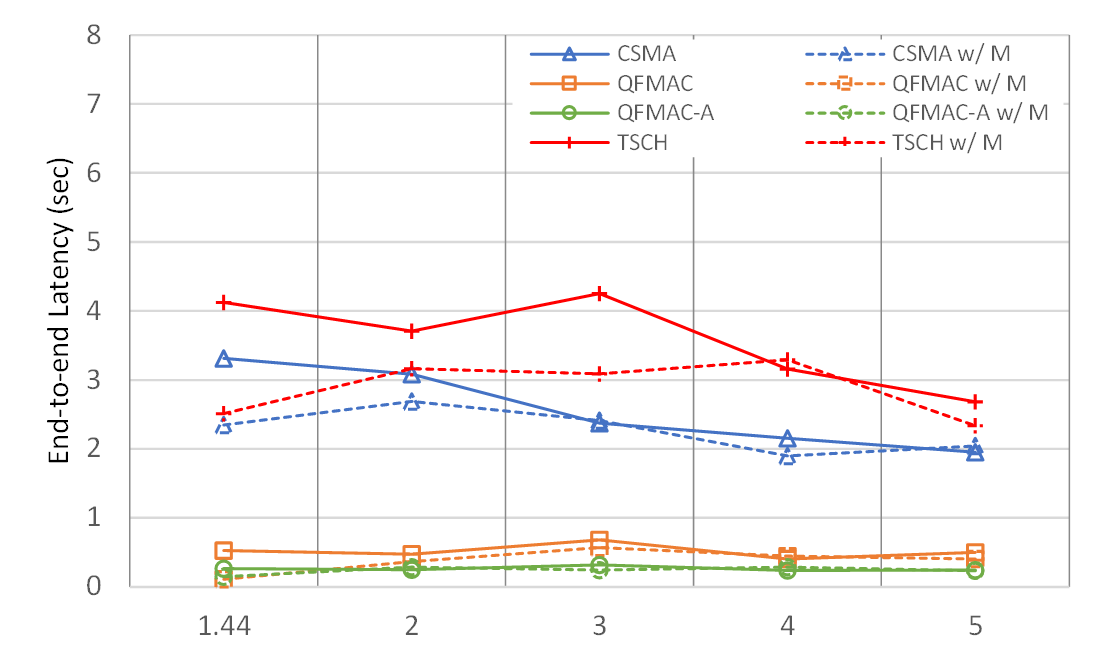}
        \label{fig_latency_unjamming}
    }
    \\
    \subfigure[Average Packet Reception Ratio]
    {
        \includegraphics[width=0.39\textwidth]{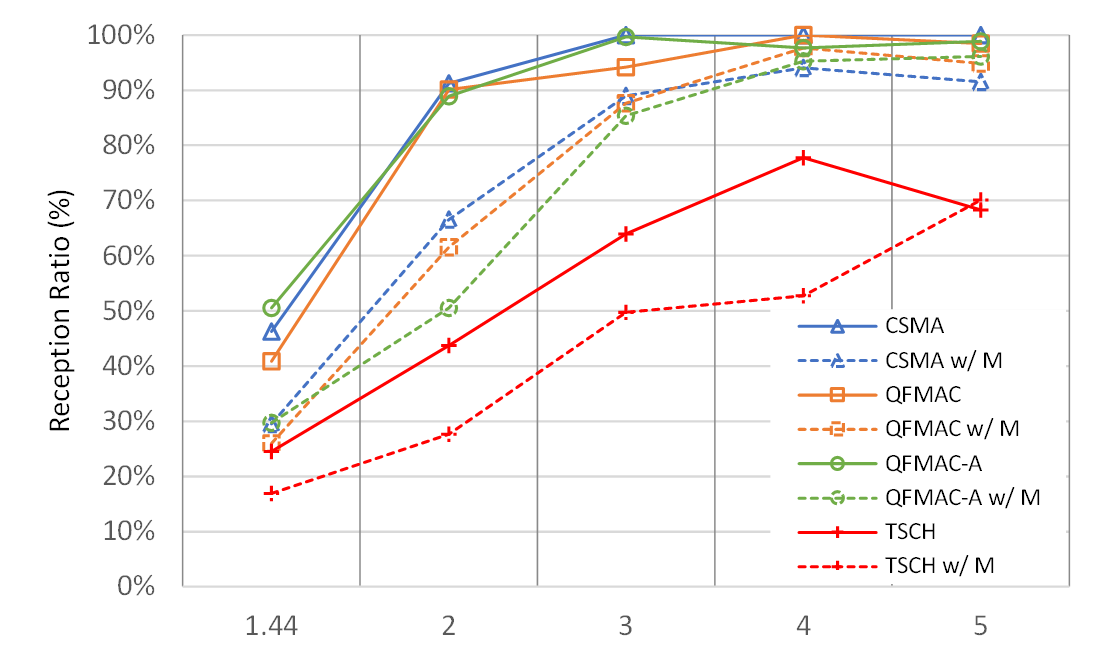}
        \label{fig_rr_unjamming}
    }
    \subfigure[Average Overhead]
    {
        \includegraphics[width=0.39\textwidth]{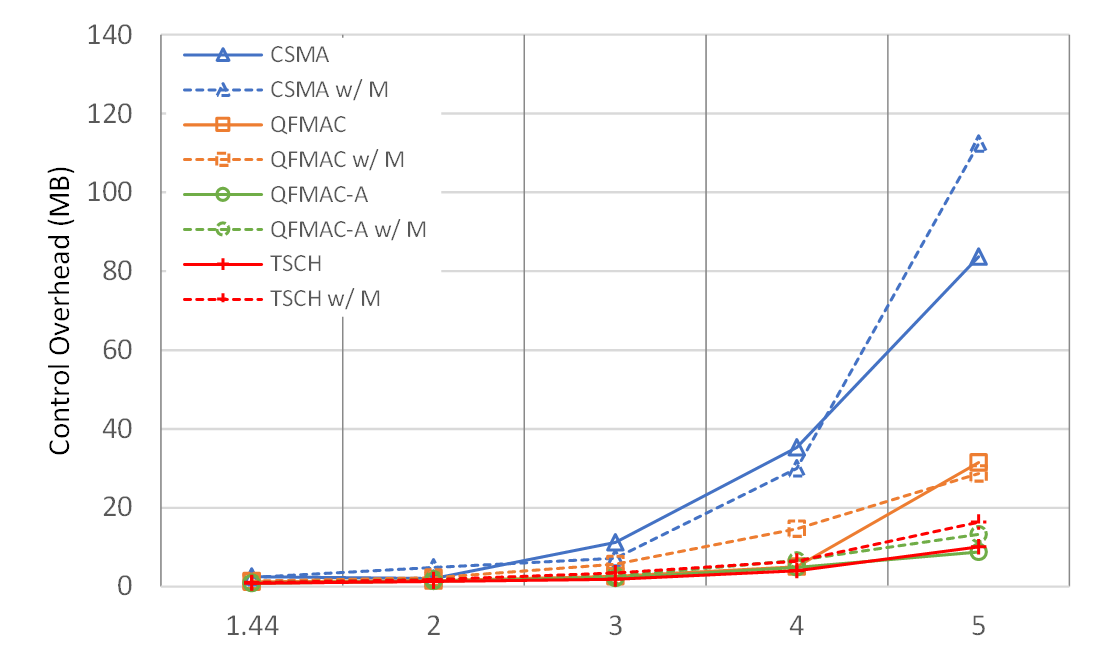}
        \label{fig_overhead_unjamming}
    }
    \caption{Density versus different performance metrics for QuickFire MAC versus CSMA and TSCH at network size of 125. The solid and dashed (denoted with ``w/ M'') lines respectively denote network scenarios that are static and mobile.}
    \label{fig_metrics_unjamming}
\end{figure*}

\begin{figure*}[hbt!]
    \centering
     \subfigure[64 nodes]    
     {
        \includegraphics[width=0.39\textwidth]{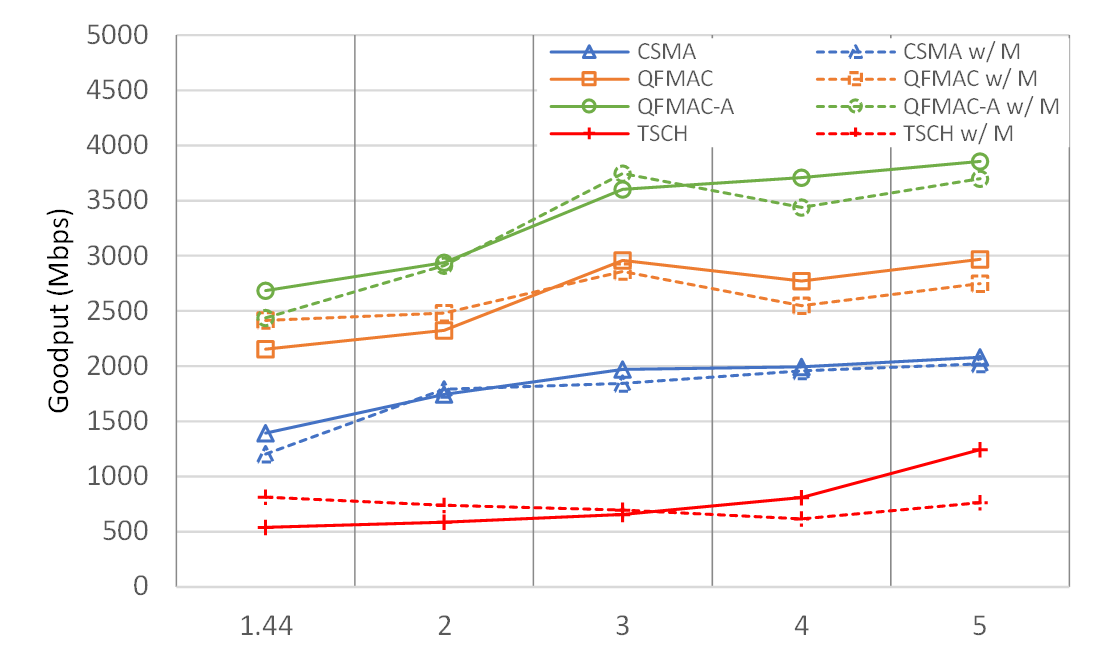}
        \label{fig_goodput_unjamming_size64}
    }
    \subfigure[216 nodes]
    {
        \includegraphics[width=0.39\textwidth]{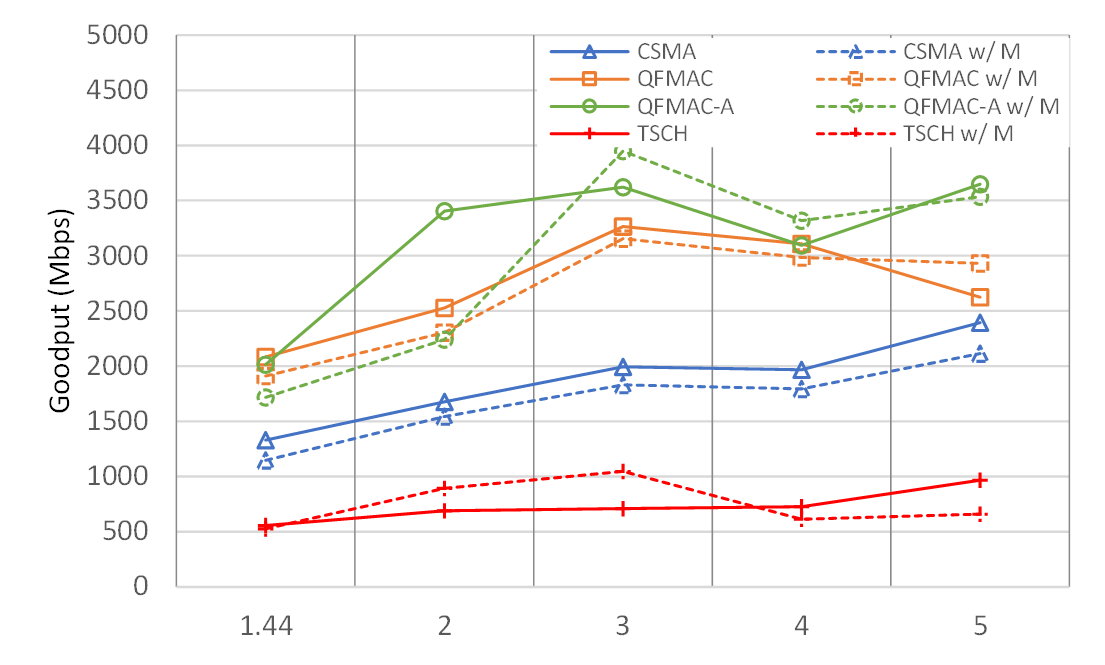}
        \label{fig_goodput_unjamming_size216}
    }
    \caption{Density versus Goodput for QuickFire MAC versus CSMA and TSCH at network sizes of 64 and 216.}
    \label{fig_goodput_unjamming}
\end{figure*}

\subsection{Performance in Static and Mobile Networks}
We begin by comparing goodput performance:  Figs.~\ref{fig_goodput_unjamming_size125} and \ref{fig_goodput_unjamming} show that QF-MAC versions substantially outperform TSCH and CSMA over all configurations of density, size, and mobility. CSMA, even operating with two wideband data channels of 3.5Mbps, does not yield high goodput efficiency because of frequent transmission delays of concurrent transmissions, which can be inferred from Fig.~\ref{fig_latency_unjamming}. TSCH has the lowest (and sub-1Mbps) goodput across almost all configurations, largely due to frequent retransmissions during concurrent communications; recall that with a global shared channel sequence, there are continuous collisions over time when $\Phi(i) \! > \! L$. This indicates that TSCH cannot sustain reliable communication in networks when there are many concurrent flows. Notably, QF-MAC with channel adaptation (QF-MAC-A) offers a remarkable improvement over the version without adaptation (QF-MAC) for both static and mobile network scenarios. With respect to network density, from $\sqrt[3]{3}$ to 2, QF-MAC and QF-MAC-A deliver more goodput as the density increases, but the goodput tends to decrease as the density increases 3 or 4 over different network sizes because of increasing impacts of inter-flow interference. We observe that mobility does not always hurt goodput: on occasion, it offers alternative paths that better avoid inter-flow interference, in which case goodput can be higher than in the static case.

With respect to end-to-end latency comparisons, Figs.~\ref{fig_latency_unjamming} likewise show that CSMA and TSCH have lower performance than QF-MAC and QF-MAC-A, again due to the frequent transmission delay of CSMA and the larger number of transmissions retries of TSCH. QF-MAC-A likewise better mitigates the impact of inter-flow interference by swapping out a transmission channel with the worst performance. In packet reception ratio comparisons, Figs.~\ref{fig_rr_unjamming} demonstrates that the use of 1Mbps channels in QF-MAC versions achieves similar reliability compared to the use of 3.5Mbps channels in CSMA. This is due to the effective avoidance of intra-flow interference and the reduction of inter-flow interference in QF-MAC. Again, the low reception ratio of TSCH is caused by frequent collisions among concurrent flows. 
We observe a sizeable gap in packet reception ratio between the static and mobile cases: this is because mobility renders the routing state stale and thus undermines the reliability of path routing.

With respect to control overhead comparisons of the MAC and routing stacks, Fig.~\ref{fig_overhead_unjamming} shows that even though CSMA incurs no MAC communication overhead to configure transmission channels, CSMA yields the largest control overhead across all densities, both with and without mobility. QF-MAC-A performs comparably to TSCH. We note that, in dense networks, arriving flows wait longer to be served as the completion time increases in the flows currently served. This is because each node serves at most four flows in a first-come-first-served fashion. As a result, unserved flows tend to perform a higher number of route exploration rounds in CSMA and TSCH, which also contributes to a higher end-to-end latency than that QF-MAC/QF-MAC-A, as seen in Fig.~\ref{fig_latency_unjamming}. In contrast, QF-MAC-A efficiently leverages the trade-off between the added control overhead for channel adaptation and the goodput efficiency to achieve substantially better goodput and end-to-end latency.. 


\begin{figure*}[hbt!]
    \centering
    \subfigure[Average Goodput]
    {
        \includegraphics[width=0.39\textwidth]{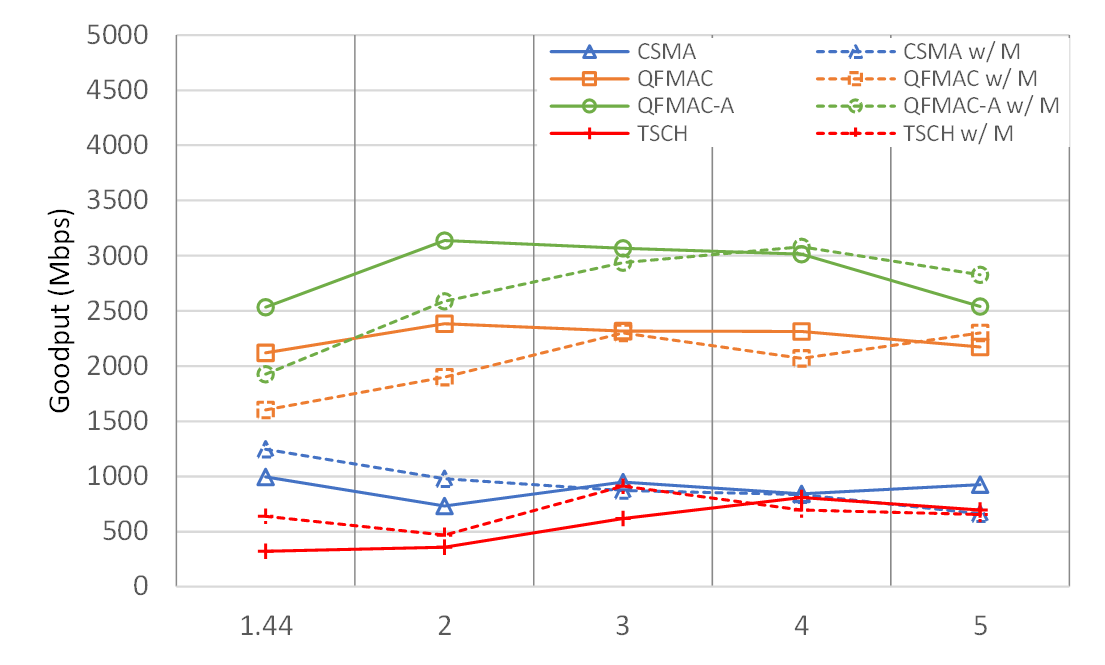}
        \label{fig_goodput_jamming_size125}
    }
    \subfigure[Average Latency]
    {
        \includegraphics[width=0.39\textwidth]{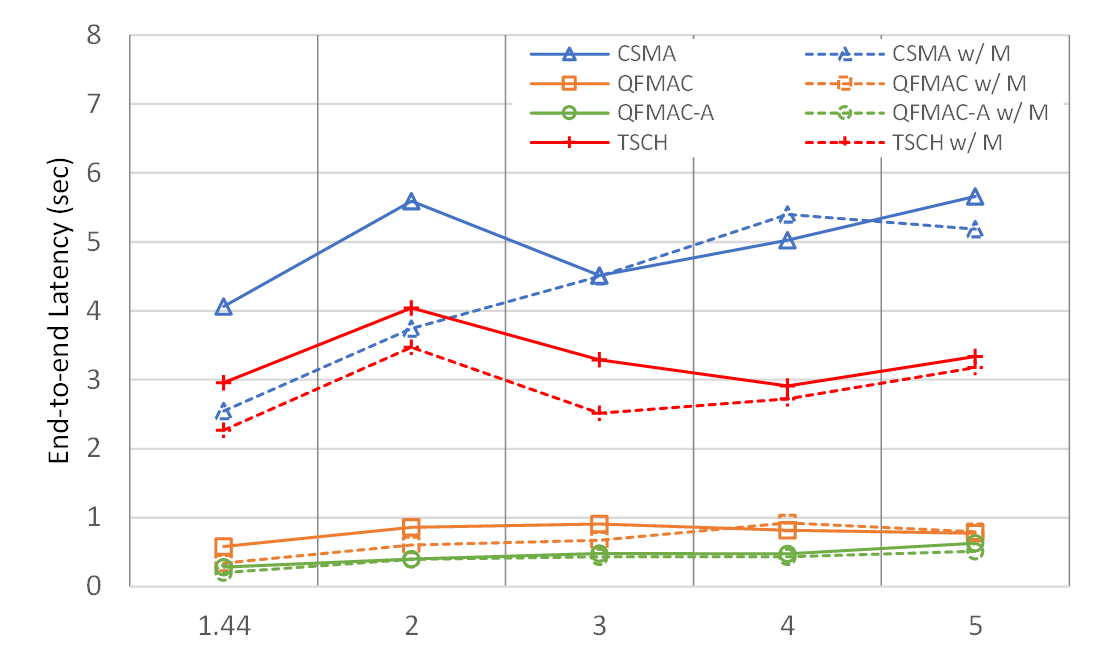}
        \label{fig_latency_jamming}
    }
    \subfigure[Average Packet Reception Ratio]
    {
        \includegraphics[width=0.39\textwidth]{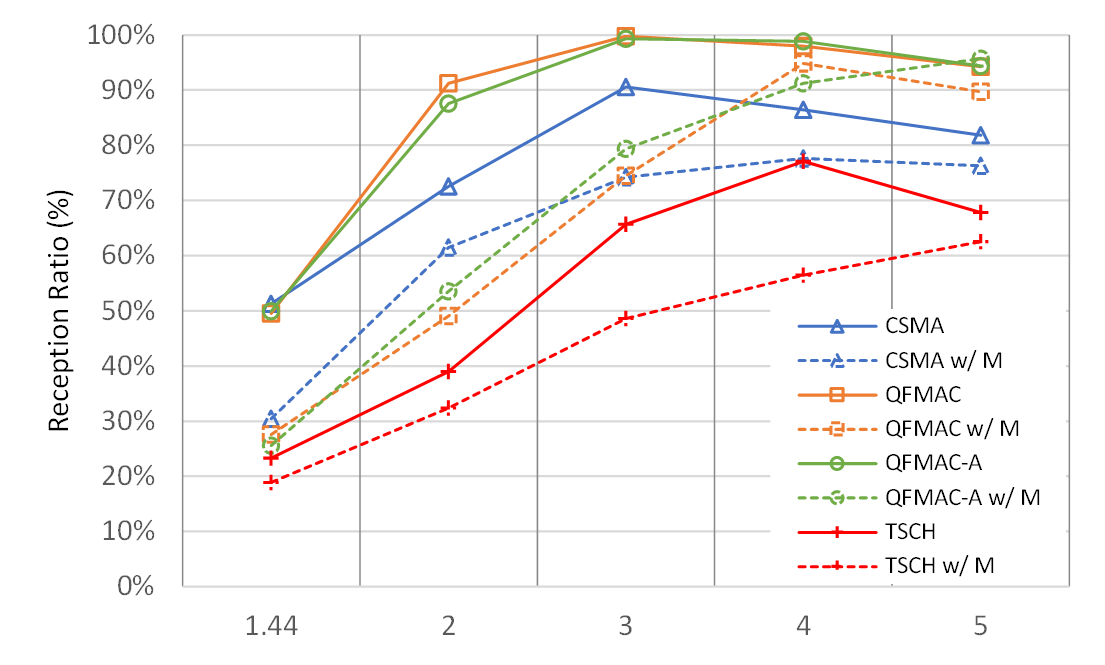}
        \label{fig_rr_jamming}
    }
    \subfigure[Average Overhead]
    {
        \includegraphics[width=0.39\textwidth]{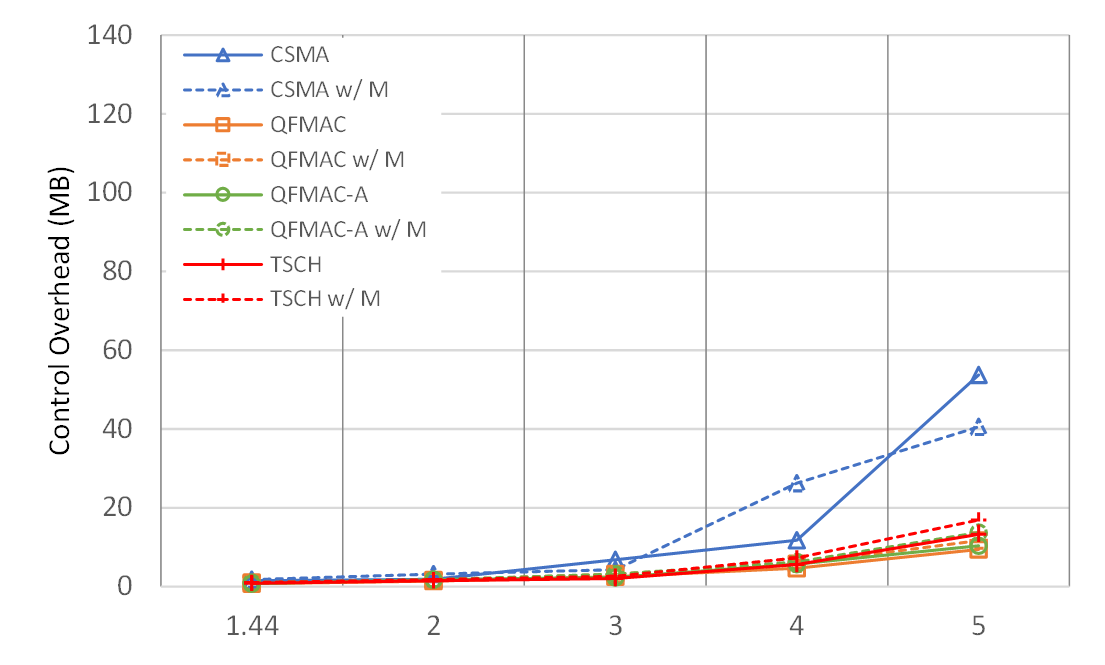}
        \label{fig_overhead_jamming}
    }
    \caption{Density versus different performance metrics for QuickFire MAC versus CSMA and TSCH at network size of 125 in the presence of adversarial jamming.}
    \label{fig_metrics_jamming}
\end{figure*}

\begin{figure*}[hbt!]
    \centering
     \subfigure[64 nodes]    
     {
        \includegraphics[width=0.39\textwidth]{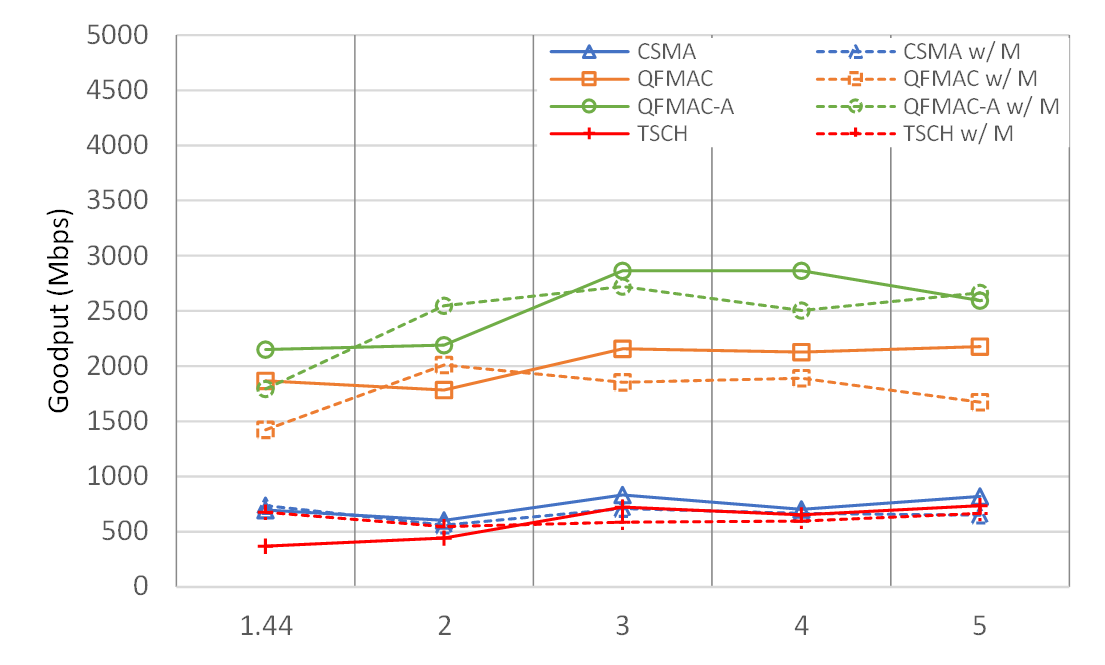}
        \label{fig_goodput_jamming_size64}
    }
    \subfigure[216 nodes]
    {
        \includegraphics[width=0.39\textwidth]{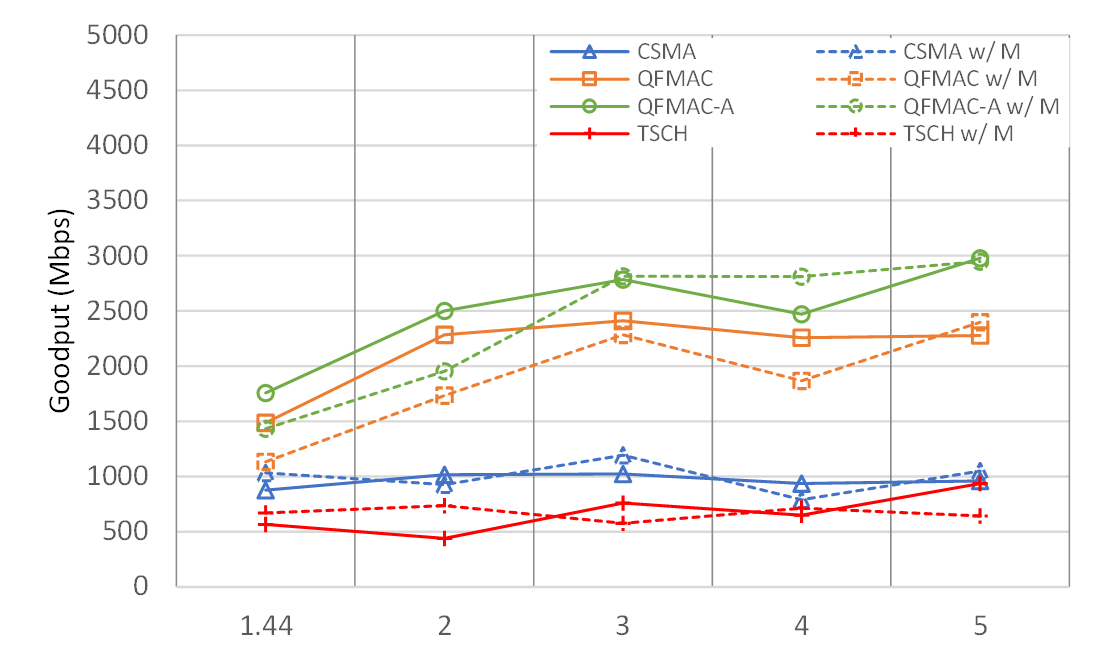}
        \label{fig_goodput_jamming_size216}
    }
    \caption{Density versus Goodput for QuickFire MAC versus CSMA and TSCH at different network sizes of 64 and 216 with adversarial jamming.}
    \label{fig_goodput_jamming}
\end{figure*}

\subsection{Performance with respect to Adversarial Jamming}
Recall that the adversarial jammer reduces the capacity of the network within its interference region by 2Mbps, in QF-MAC and TSCH, jamming two channels that have the highest aggregate traffic and, in CSMA, one of the two data channels. Figs.~\ref{fig_goodput_jamming_size125} and \ref{fig_goodput_jamming} show that QF-MAC versions effectively eschew the jammed channels and substantially outperform CSMA and TSCH across varying densities, sizes, and mobility speeds. Observe that QF-MAC and QF-MAC-A have respectively {\footnotesize ${\sim}$}20\% and {\footnotesize ${\sim}$}25\% goodput reduction, whereas CSMA has a {\footnotesize ${\sim}$}52\% drop of goodput. This implies that both the use of local channel sequences and the use of channel adaptation help mitigate the loss of achievable capacity in the presence of adversarial jamming. TSCH, with a ~10\% goodput reduction, again shows poor performance of sub-1Mbps because of its lack of adaptivity to internal and external interference.

In end-to-end latency comparisons, Fig.~\ref{fig_latency_jamming} shows that CSMA and TSCH have a much higher latency than QF-MAC and QF-MAC-A in the presence of jamming. In contrast, QF-MAC versions achieve a sub-1 second latency over all configurations due to the efficient channel utilization over the channels free from jamming. 

Regarding reception ratio comparisons, Figs.~\ref{fig_rr_unjamming} and \ref{fig_rr_jamming} show that QF-MAC and QF-MAC-A are still able to sustain reception ratio under adversarial jamming. CSMA suffers an obvious reliability decrease of over 10\% once the density crosses 3, with and without mobility. Since the effect of inter-flow interference from concurrent flows yields numerous retransmissions in TSCH, it again achieves a poor reception radio in the presence of external jamming. With respect to control overhead comparisons, Fig.~\ref{fig_overhead_jamming} shows essentially the same overhead comparison as in the case without adversarial jamming.\footnote{Surprisingly, the control overhead of QF-MAC-A and CSMA decreases when jamming is introduced, especially once the density increases to 3 or more, but this is due to a subtle artifact of the routing protocol, which chooses the minimum interference path in the data plane and essentially finds paths that circumvent the jammed region of the network. A detailed explanation of this reduction is outside the scope of this paper, as this pattern is not seen for other routing protocols.}

\section{Conclusion}
By leveraging a local channel hopping schedule for multi-hop communication, QF-MAC adapts across diverse network scenarios (e.g., of density, concurrency, mobility, etc.) and is robust to internal and external interference in multi-radio networks. In contrast, TSCH-based protocols, which rely on a link offset schedule to stagger a global channel hopping sequence, are prone to more collisions or otherwise incur goodput-latency performance penalties. Similarly, CSMA-based protocols, albeit more robust to collisions than TSCH, underperform in goodput-latency and thus have relatively high control overhead as well. By allowing nodes to leverage local interference estimation to further adapt their channels, QF-MAC achieves a gain in channel efficiency with low control overhead, even in the presence of high interference dynamics. QF-MAC works with on-demand routing over single or multiple paths and is also compatible with table-driven routing protocols. 

In future work, we plan to study how to take advantage of deep learning and reinforcement learning to perform accurate interference prediction with lightweight sampling overhead locally.  We will also consider techniques to deal with the jamming of the control channel itself and further explore multi-radio goodput efficiency enhancement.

\bibliographystyle{IEEEtran}
\bibliography{myref}

\newpage

\appendices

\section{Analysis of Configuration Space and Performance in Channel Hopping Sequences}
\label{FirstAppendix}

This section analyzes the configuration space of the network and scheduling parameters and compares the performance across different scheduling of channel hopping sequences. The objective of the analysis is to understand the vulnerability of TSCH and explore the feasibility of other scheduling schemes to improve the average throughput in consideration of channel collision across radios associated with the same or different flows within an interference region. We again note that, in an interference region, the euclidean distance between any two radios is not greater than both the interference radius $R_i$ and the transmission radius $R_t$. Therefore, any two radios are able to interfere with each other during data transmission.

\subsection{Configuration space of channel hopping scheduling}
We let $U$ denote the number of all channels and $M$ denote the number of radios in a given interference region. Note that $U \geq M$ in Eq.~\ref{SEQ_LEN} so that all radios get a chance to use disjoint channel at a slot. Within this range, $F$ active flows deliver data packets at each slot. Since the network enables concurrent transmissions, each radio associated with a flow keeps sending packets at each slot. Without loss of generality, we let $K = \frac{M}{F}$ represent the number of radios in a flow within the interference region. $L$ defines the length of a cyclic sequence of frequencies. Note that each frequency disjoints in a sequence if $L \leq U$. When $L > U$, the sequence is broken down into $\lceil \frac{L}{U} \rceil$ segments, where up to $U$ frequencies are disjoint.

\subsection{Scheduling of channel hopping sequences}
\textbf{Random channel hopping sequence.} Each radio locally generates a cyclic sequence of frequencies without coordination. Channel collision can happen for any two radios associated with either the same flow or different flows.

\textbf{Global channel hopping sequence.} TSCH applies a globally fixed cyclic sequence of frequencies shared by all radios. It relies on channel offset schedules to avoid channel collision among different transmitting radios within an interference region.

\textbf{Fixed per-flow channel hopping sequence.} All radios in a flow share a cyclic sequence of frequencies. The sequence is randomly chosen from the set of channels at a flow setup. Within each flow, the channel offset of the next forwarder is right-shifted by one to avoid intra-flow channel collision.

\subsection{Analysis of collisions/slot}
In order to understand the survivability of sequence schedulers in a network without additional control overhead across data flows, all sequence schedulers are performed without inter-flow coordination. Thus, inter-flow channel collision may occur because a radio is unaware of the radios' sequence information of other flows. On the other hand, intra-flow channel collision can be avoided by simply shifting the channel offset by one position. This offset shifting can be done at a flow setup, which is to select a routing path and assign the transmission/reception channel sequences for a flow.

In a scheduler with random sequences, the sequence assignment at each radio is independent. Thus, both intra-flow and inter-flow channel collisions may occur. At radio $v$, the probability of not getting collided with one other radio $u$ at a slot is $1 - \frac{1}{U}$. Since there are $M-1$ radios potentially interfering with the channel used by radio $v$, the probability of collision per slot at a radio is calculated as follows:

\begin{equation}
\begin{array}{l}
\label{collision_probability_random}
P_{random} = 1 - (\frac{U-1}{U})^{M-1}
\end{array}
\end{equation}

In a scheduler with a global channel hopping sequence, the offset is randomly chosen at the source node at the flow setup. Then the offset is added by one whenever the flow control message is propagated to the next hop in a downstream direction. For each flow with radios $v_1, ..., v_K$, there are $K$ disjoint channels selected for $v_1, ..., v_K$ at a slot. Given a selected channel associated with radio $v_{x_1} \in$ flow $f_x$, the probability that all radios associated with other flow $f_y$ cause no conflict is $\frac{L-K}{L}$. Since there are $F-1$ flows potentially interfering with the channel used by radio $v_{x_1}$, the probability of collision per slot at a radio is calculated as follows:

\begin{equation}
\begin{array}{l}
\label{collision_probability_global}
P_{global} = 1 - (\frac{L - K}{L})^{F-1}
\end{array}
\end{equation}

In the scheduler with a fixed per-flow channel hopping sequence, similar to the global channel hopping sequence, the offset is assigned within a flow in a downstream direction to avoid intra-flow channel collision. However, compared to a global channel hopping sequence with only $L$ channels used in the interference range, the scheduler allows each flow randomly to select $L$ channels from the set of all channels. Therefore, given a selected channel associated with radio $v_{x_1} \in$ flow $f_x$, the probability that all radios associated with other flow $f_y$ cause no conflict is $\frac{U-K}{U}$. Since there are $F-1$ flows potentially interfering with the channel used by radio $v_{x_1}$, the probability of collision per slot at a radio is calculated as follows:

\begin{equation}
\begin{array}{l}
\label{collision_probability_fixed_per_flow}
P_{per flow} = 1 - (\frac{U - K}{U})^{F-1}
\end{array}
\end{equation}
Note that $P_{per flow}$ is independent of $L$. Thus, only a short sequence is needed in the scheduling of fixed per-flow channel hopping sequences. This feature improves the flexibility and diversity of chosen channels in each flow and thus saves the payload size in control messages of flow setup.

\begin{lemma} 
In an interference region with concurrent flows, where $U \geq M \geq F > 1$ and $\frac{M}{F} = K \geq 1$, scheduling fixed per-flow channel hopping sequences achieves a performance not worse than the one in scheduling global or random channel hopping sequences: $P_{per flow} \leq P_{random}$ and $P_{per flow} \leq P_{global}$.
\end{lemma}

\begin{proof}

In the comparison of Eq.~\ref{collision_probability_global} and Eq.~\ref{collision_probability_fixed_per_flow}, $P_{per flow}$ is not greater than $P_{global}$ because of $U \geq L$ and $\frac{U - K}{U} \geq \frac{L - K}{L}$. Namely, scheduling fixed per-flow channel hopping sequences yields an average throughput not worse than a global channel hopping sequence within an interference region.


In the comparison of Eq.~\ref{collision_probability_random} and Eq.~\ref{collision_probability_fixed_per_flow}, to show that $(\frac{U - 1}{U})^{M-1} \leq  (\frac{U - K}{U})^{F-1}$,we consider three cases: (i) $U \geq M = F$, (ii) $U = M \geq F$, (iii) $U > M > F$.

\textbf{Case i}: $U \geq M = F, K=1$. we have $
(\frac{U - 1}{U})^{M-1} = (\frac{U - K}{U})^{F-1}$

\textbf{Case ii}: $U = M \geq F$. we have$
(\frac{U - 1}{U})^{M-1} =  (\frac{U - 1}{U})^{U-1}$ and
$(\frac{U - K}{U})^{F-1} = (\frac{U - \frac{U}{F}}{U})^{F-1} = (\frac{F-1}{F})^{F-1}$. Because $(\frac{x - 1}{x})^{x-1}$ is decreasing when $x \geq 1$, $(\frac{U - 1}{U})^{U-1} < (\frac{F-1}{F})^{F-1}$ holds if $U > F \geq 1$. Therefore, we have $(\frac{U - 1}{U})^{M-1} < (\frac{U - K}{U})^{F-1}$.

\textbf{Case iii}: $U > M > F$. We first show that the inequality of $(\frac{U - 1}{U})^{M-1} < (\frac{U - K}{U})^{F-1}$ can be derived to the following inequality:

\begin{equation}
\begin{array}{l}
\label{inequality_1}
\ \ \ \ \ \ \ (\frac{U - 1}{U})^{M-1} < (\frac{U - K}{U})^{F-1} \\
\iff (1-\frac{1}{U})^{\frac{M-1}{F-1}} < 1- \frac{K}{U} = 1- \frac{M}{UF}
\end{array}
\end{equation}

\begin{equation*}
\begin{array}{l}
\ \ \ \ \ \ \ 0 < (\frac{U - 1}{U})^{M-1} < (\frac{U - K}{U})^{F-1} < 1 \\
\iff \log{(\frac{U - 1}{U})^{M-1}} < \log{(\frac{U - K}{U})^{F-1}} \\
\iff (M-1)(\log{(U-1)} - \log{U}) < \\
\ \ \ \ \ \ \ (F-1)(\log{(U-K)} - \log{U}) \\
\iff \frac{M-1}{F-1}(\log{(U-1)} - \log{U}) < log{(U-K)} - \log{U} \\
\iff \frac{M-1}{F-1} > \frac{\log{(U-K)} - \log{U}}{\log{(U-1)} - \log{U}} \\
\iff \frac{M-1}{F-1} > \frac{\log{U} - \log{(U-K)}}{\log{U} - \log{(U-1)}}\\

\iff \frac{M-1}{F-1} > \frac{\log{U} - (\log{U} + \log{(1 - \frac{K}{U}}))}{\log{U} - (\log{U} + \log{(1 - \frac{1}{U}}))} \\
\iff \frac{M-1}{F-1} > \frac{\log{(1 - \frac{K}{U})}}{\log{(1 - \frac{1}{U})}} \\
\iff \frac{M-1}{F-1} \log{(1 - \frac{1}{U})} < \log{(1 - \frac{K}{U})} \\
\iff (1 - \frac{1}{U})^{\frac{M-1}{F-1}} < 1 - \frac{K}{U} = 1- \frac{M}{UF} \\
\end{array}
\end{equation*}

Note that the exponential function of $f(M) = (1 - \frac{1}{U})^{\frac{M-1}{F-1}}$ is decreasing when $U \geq M \geq F$ and the exponential function of $f(M) = (1 - \frac{1}{U})^{\frac{M-1}{F-1}}$ and the linear function of $g(M) = 1- \frac{M}{UF}$ has at most two intersections, where one of them occurs at $ M = F$. Also, we know that $(\frac{U - 1}{U})^{M-1} < (\frac{U - \frac{M}{F}}{U})^{F-1}$ and $(1 - \frac{1}{U})^{\frac{M-1}{F-1}} < 1- \frac{M}{UF}$ hold at $U = M$, according to the derivation in Case ii and Eq.~\ref{inequality_1}. Based on the two statements that $(1 - \frac{1}{U})^{\frac{M-1}{F-1}} < 1- \frac{M}{UF}$ at $U = M$ and $(1 - \frac{1}{U})^{\frac{M-1}{F-1}} = 1- \frac{M}{UF}$ at $M = F$, we can derive that $(1 - \frac{1}{U})^{\frac{M-1}{F-1}} < 1- \frac{M}{UF}$, for $M \in (F, U]$. This is because that $f(M) = (1 - \frac{1}{U})^{\frac{M-1}{F-1}}$ is decreasing and the possibly remaining one intersection of $f(M)$ and $g(M)$ will not occur in the range of $(F, U)$. (Otherwise, $(1 - \frac{1}{U})^{\frac{M-1}{F-1}} > 1- \frac{M}{UF}$.) Hence, $(\frac{U - 1}{U})^{M-1} < (\frac{U - K}{U})^{F-1}$, for $U > M > F$, is proved.  

According to the three cases above, $(\frac{U - 1}{U})^{M-1} \leq  (\frac{U - K}{U})^{F-1}$, for $U \geq M \geq F > 1$, and then $P_{per flow} \leq P_{random}$  is proved.

\end{proof}

\end{document}